\documentclass{ocg}
\usepackage{amsfonts}
\usepackage{amsmath}
\usepackage{latexsym}
\usepackage{amssymb}
\usepackage{tikz}
\usetikzlibrary{matrix}

\theoremstyle{plain}
\newenvironment{soundness}{\paragraph{Soundness:}}{}

\newcommand{\VS}[1]{#1} 	
\newcommand{\VL}[1]{} 	

\newcommand{\cim}[1]{}
\newcommand{\couic}[1]{}
\newcommand{\ports}{\!\!:\!\!}

\newcommand{\reconsider}[1]{\textcolor{gray}{}}

\begin{document}

\title{Causal dynamics of discrete manifolds}
\author[1]{Pablo Arrighi}
\author[3]{Cl\'ement Chouteau}
\author[1]{Stefano Facchini}
\author[4]{Simon Martiel}

\address[1]{Aix-Marseille Univ., Universit\'e de Toulon, CNRS, LIS, Marseille, and IXXI, Lyon, France\\
\email{\{pablo.arrighi,stefano.facchini\}@univ-amu.fr}}
\address[3]{Inria, LSV, ENS Paris-Saclay 61, avenue du Pr\'esident Wilson, 94235 Cachan Cedex, France}
\address[4]{Atos/Bull, Quantum R\&D, 78340 Les Clayes-sous-Bois, France\\
\email{simon.martiel@atos.net}}

\maketitle

\begin{abstract}
We extend Cellular Automata to time-varying discrete geometries. In other words we formalize, and prove theorems about, the intuitive idea of a discrete manifold which evolves in time, subject to two natural constraints: the evolution does not propagate information too fast; and it acts everywhere the same. For this purpose we develop a correspondence between complexes and labeled graphs. In particular we reformulate the properties that characterize discrete manifolds amongst complexes, solely in terms of graphs. In dimensions $n<4$, over bounded-star graphs, it is decidable whether a Cellular Automaton maps discrete manifolds into discrete manifolds.\medskip\\
\end{abstract}


\section{Introduction}

Discrete geometry refers to discretizations of continuous geometries, i.e. piecewise-linear manifolds, that can be abstracted as combinatorial objects such as simplicial complexes, etc. But it may also refer to mere graphs/networks equipped with their natural graph distance. This ambiguity is common in Computer Science, but also in Physics. For instance in discrete/quantized versions of General Relativity, spacetime is discretized as simplicial complexes (Regge-calculus) or in the basis of spin networks graphs (Loop Quantum Gravity). This raises the question of a thorough comparison between simplicial complexes, and graphs. 

A natural way of approaching this question is to seek to encode complexes into labeled graphs. Then, a natural way to encode a complex into a labeled graph, is to map: each simplex $u$ into a vertex $u$; each facet $u\ports a$ of the simplex into a port $u\ports a$ of the vertex $u$; each gluing between facets $u\ports a$ and $v\ports b$ into an edge $(u\ports a,v\ports b)$; each possible way of rotating/articulating this gluing as a label $\gamma$ carried by this edge---see Fig. \ref{fig:correspondance} and  \ref{fig:permutations}.  We formalize this correspondence in Section \ref{sec:complexesasgraphs}. Notice that by `complex', we really mean `pseudo-manifold' here, i.e. each facet of a simplex is attached to one other facet at most. Notice also that the precise choices of ports may not matter, so long as the edges between them represent oriented gluings of simplices. There is, therefore, a local rotation symmetry.

This works well, but a non-often emphasized problem arises. Consider triangles hinging around a point, as in the bottom left of Fig. \ref{fig:correspondance}. The geometrical distance between the two extreme tetrahedrons is one, since they share a point. But the graph distance between their corresponding vertices is three, and the path between them could be made---they are not graph neighbors. There is a discrepancy between the two distances, which we characterize in Section \ref{sec:distances}. Faced with this discrepancy, we have two options.

One option is to forget about geometrical distance altogether. Indeed, if one thinks of each tetrahedron as a room (as in the movie Cube, say), then it is the graph distance that matters. In Section \ref{sec:CDC}, we develop a theory of Causal Dynamics of Complexes (CDC). CDC evolve complexes in discrete time steps, subject to two natural constraints: the evolution does not propagate information too fast; and it acts everywhere the same. This is thanks to the concept of Causal Graph Dynamics (CGD), which we recall. We prove that the CGD which commute with local rotations, can always be implemented with rotation-commuting local rules---a property which in turn is decidable. Then, the previously developed correspondence between complexes and labeled graphs readily allows us to reinterpret these rotation-commuting CGD, as CDC. CDC are already interesting as a mathematically rigorous framework in which to cast the more pragmatical simplicial complex parallel rewrite system of \cite{GiavittoMGS}, or in order to explore causal dynamics of Causal Dynamical Triangulations \cite{LollCDT}, \`a la \cite{MeyerLove}. 

Another option is to take geometrical distance into account. Then, looking at the complex at this larger scale, and in dimension $3$ and above, unravels new concerns. For instance  Fig. \ref{fig:torsion} has the topology of a pinched ball (think of a balloon compressed between two fingers until they touch each other). This is not a manifold, since the neighborhood of the compression point is not a ball. Thus, since the cycle length is arbitrary, the property of being a manifold, is non-local. To make matters worse, the two extreme tetrahedrons could have been glued in a torsioned manner, see Fig \ref{fig:torsion} again. A somewhat radical solution to these concerns is to restrict to complexes such that, even in the geometrical distance, neighborhoods are bounded. Then, the discrepancy between the geometrical and the graph distances is linearly bounded. Another motivation for considering these `bounded-star complexes' is if the next state of a tetrahedron is computed from that of the geometrically neighboring tetrahedrons: we may want this neighborhood to be bounded, whether for practical purposes (e.g. efficiency of a finite-volume elements methods) or theoretical reasons (e.g. computability from the finiteness of the local update rule; finite-density as a physics postulate). Finally, this is a way to prevent sudden geometrical distance collapse---see Fig. \ref{fig:boundedstar}. 

In Section \ref{sec:pachner} we characterize manifolds. In continuous geometries, a manifold is characterized by the neighborhood of every point being homeomorphic to a ball. Over simplicial complexes, this translates into the neighborhood of every simplex being homeomorphic to a ball, where the notion of homeomorphism is captured by a finite set of local rewrite rules, often referred to as Pachner moves (technically, bistellar moves plus shellings and their inverses). These moves are reformulated in terms of graph moves; but the obtained graph moves are not graph-local. We show that a subset of these graph moves is just as expressive as the Pachner moves, whilst enjoying the property of being graph-local. That way, the properties that characterize discrete manifolds are reformulated in terms of graph-local moves. 

In Section \ref{sec:CDDM}, we show that it is decidable whether a CGD is torsion-free bounded-star discrete-manifold preserving---in dimensions less than four. Then, the correspondence between complexes and labeled graphs readily allows us to reinterpret these, as Causal Dynamics of Discrete Manifolds (CDDM).  

We conclude in Section \ref{sec:conclusion} with a discussion of the result and their connection with past and future works, as well as a detailed comparison with the crystallizations/gems alternative. 

\section{Complexes as graphs}\label{sec:complexesasgraphs}

\begin{figure}\begin{center}
\includegraphics[scale=0.4]{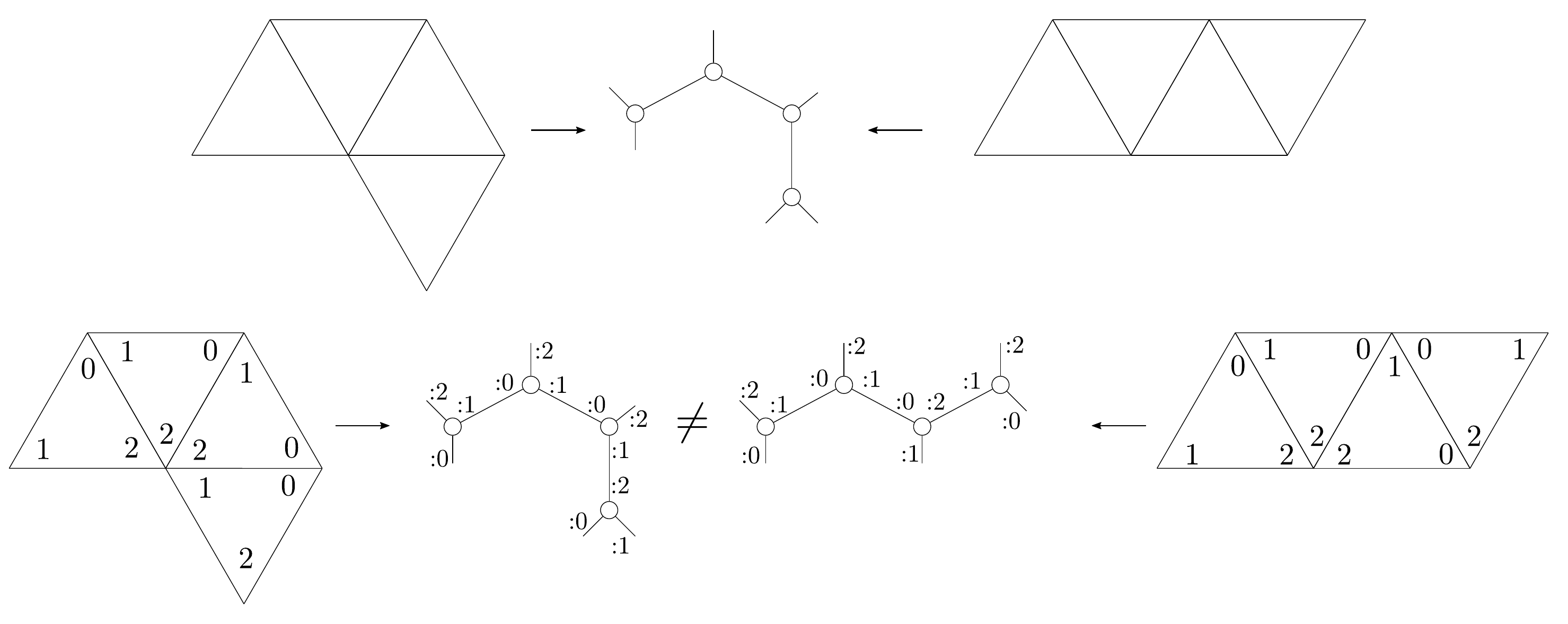}
\end{center}\vspace{-0.5cm}
\caption{Complexes as graphs. \label{fig:correspondance}
{\em Top row}. The naive way to encode complexes as graphs is ambiguous.
{\em Bottom row}. Encoding colored complexes instead lifts this ambiguity. }
\end{figure}

The naive way is to map each simplex to a vertex $v$, and each gluing between facets to an edge $\{u,v\}$. The problem, then, is that we can no longer tell one facet from another, which leads to ambiguities (see Fig. \ref{fig:correspondance} {\em Top row.}). A first solution attempt is to consider {\em colored simplicial complexes} instead. In these complexes, each of the $n+1$ points of a $n$--simplex has a different color. Now we can map each simplex to a vertex, and each gluing between facets to an edge, but now this edge $\{u\ports p,v\ports q\}$ holds the colors of the points that are opposite the glued facets (see Fig. \ref{fig:correspondance} {\em Bottom row.}). The problem, now, is that as soon as we consider $3$--dimensional complexes, there are three different, rotated/articulated ways of gluing two tetrahedrons along two given facets (see Fig. \ref{fig:permutations}). We must therefore provide a permutation $\gamma$ telling us which points identifies with whom, on the edges. Because these permutations are not, in general, involutions, we must direct our edges. Odd permutations correspond to oriented gluings. Altogether this leads to the following definition, which is an elaborated version of \cite{ArrighiCGD,ArrighiIC,ArrighiCayleyNesme}:

\begin{figure}[!h]
\begin{center}
\includegraphics[scale=1.0]{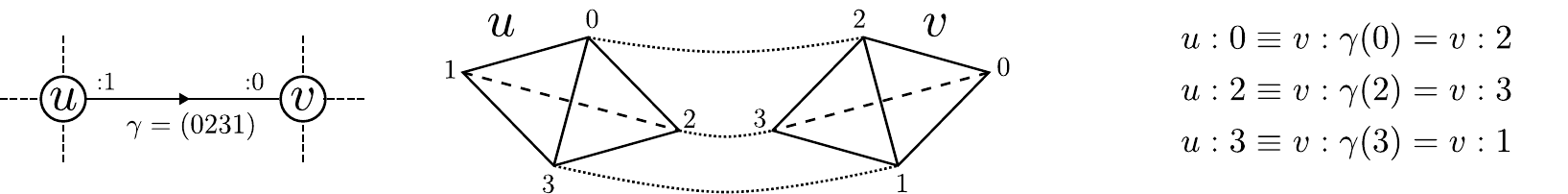}
\end{center}
\caption{The different, rotated/articulated ways of gluing of two tetrahedrons along two given facets are specified on the edges. \label{fig:permutations}}
\end{figure}

\begin{definition}[Graphs, Disks]\label{def:graphs}
Let $V$ be an infinite countable set referred to as the `universe of names'. Let $\Sigma$ be a finite set referred to as the `internal states'. Let $n$ stand for the spatial dimension. Let $\pi$ be ${0\ldots\ n+1}$ referred to as the `set of  ports'. Let $\Gamma$ be the $(n+2)!/2$ odd permutations of $\pi$, referred to as `gluings'. A {\em graph} $G$ is given by 
\begin{itemize}
\item[$\bullet$] A subset $V(G)$ of $V$---whose elements are called vertices.
\item[$\bullet$] A function $\sigma : V(G) \rightarrow \Sigma$ associating to each vertex its label.
\item[$\bullet$] A set $S(G)$ of elements of the form $(u\ports p)$ with $u\in V(G)$, $p\in \pi$---whose elements are called semi-edges.
\item[$\bullet$] A set $E(G)$ of elements of the form $(u\ports p, \gamma,v\ports q)$ with $u, v \in V(G)$, $p,q \in \pi$, $\gamma\in \Gamma$---whose elements are called edges.
\end{itemize}
This is with the conditions that
\begin{itemize}
\item if $(u\ports p)\in S(G)$ then there is no $(u\ports p, \gamma,v\ports q) \in E(G)$.
\item if $(u\ports p, \gamma,v\ports q) \in E(G)$ then there is no other $(u\ports p, \gamma',v'\ports q') \in E(G)$.
\item each vertex has exactly $n+1$ ports, i.e. appearing in $S(G)\cup E(G)$.
\item if $(u\ports p, \gamma,v\ports q) \in E(G)$ then $\gamma$ must map the $n+1$ ports of $u$ into the $n+1$ ports of $v$, with $\gamma(p)=q$.
\item if $(u\ports p, \gamma,v\ports q) \in E(G)$ then $(v\ports q, \gamma^{-1}, u\ports p) \in E(G)$.
\end{itemize}
The set of graphs is denoted $\mathcal{G}$. Given a graph $G$, we write $G^r_v$ for its disk of radius $r$ centered on $v$, i.e. its subgraph induced by those vertices that lie at graph distance less or equal to $r+1$ from $v$ in $G$, breaking outgoing edges into semi-edges. The set of disks of radius $r$ is denoted $\mathcal{D}^r$.
\end{definition}

\begin{figure}[h]\begin{center}
\includegraphics[scale=1.0]{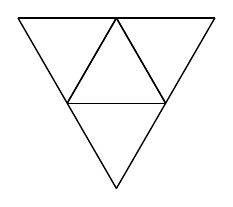}\includegraphics[scale=1.0]{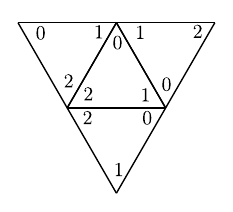}\includegraphics[scale=1.0]{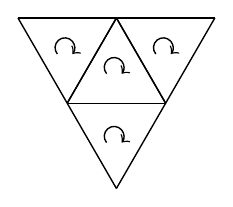}
\caption{Complexes, Colored complexes, Oriented Complexes \label{fig:complexes}}
\end{center}\end{figure}

These graphs correspond to colored complexes, i.e. gluings of simplices whose points have colors. 
Colored simplicial complex are not uncommon, but certainly not as common as oriented complexes, however---see Fig. \ref{fig:complexes}. If we wish to remove colors, we must allow for the simplices to rotate freely. On graphs this corresponds to reshuffling the ports in $\pi$ according to an even permutation, i.e. a rotation. 
\begin{definition}[Vertex rotations and symmetries]
Let $G$ be a graph, $u$ one of its vertex and $r$ an even element of $\Pi$. Then, a {\em vertex rotation} $r_u$ is the application of $r$ at $u$. More precisely, $G'=(r_u)G$, is such that
\begin{itemize}
\item $V(G')=V(G)$.
\item $E(G^\prime)$ and $S(G^\prime)$ are the image of $E(G)$ and $S(G)$ under the map:
\begin{itemize}
\item $(u\ports p, \gamma,v\ports q) \mapsto (u\ports r(p), \gamma\circ r^{-1},v\ports q)$.
\item $(v\ports q, \gamma^{-1},u\ports p) \mapsto (v\ports q, r\circ \gamma^{-1},u\ports r(p))$.
\item $(u\ports p) \mapsto (u\ports r(p))$.
\end{itemize}
\item $\sigma'(u)=h(r)(\sigma(u))$, whereas $\sigma'(v)=\sigma(v)$ for $v\neq u$,
\end{itemize}
where $h$ is a given homomorphism between the group of permutations $\Gamma$ over $\pi$, and a group of transformations $h(\Gamma)$ over $\Sigma$. A rotation sequence $\overline{r}$ is a finite composition $\prod r^i_{u_i}$, with $r^i$ some rotations, and $u_i$ some vertices. When $s$ is an odd element of $\Pi$, we can similarly define a {\em vertex symmetry} $s_u$ and a symmetry sequence $\overline{s}$. We use $s_{ij}$ as a shorthand notation for the flip between $i$ and $j$.
\end{definition}
\reconsider{
\begin{soundness}
We have to check that in the graph $G^\prime$ image of $G$ by $r@u$ the two conditions on the edges respected by graphs are still valid in $G^\prime$.
For the edges not linked to $u$ it is obvious.
Else the other edges in $E(G^\prime)$ are of the form $(u:r(p), \gamma^\prime, v:q)$, we have:
\begin{enumerate}
\item $(v\ports q, r\circ \gamma^{-1},u\ports r(p))$ is the inverse of $(u\ports r(p), \gamma\circ r^{-1},v\ports q)$.
\item $\gamma\circ r^{-1}(r(p)) = \gamma(p) = q$ and $r\circ \gamma^{-1}(q) = r(p)$.
\end{enumerate}
\end{soundness}
Notice that, since the composition of two odds permutation is even, applying a symmetry to a strict subset of $V(G)$ produces an object which is not a graph (\emph{ie} with even permutations on the edges). Moreover, applying two symmetries on the same vertex can be summed up as applying a rotation to the vertex. Therefore, when applying a symmetry sequence, we must ensure that symmetries are applied an even number of time on every vertex of the graph.
We can also assume that a symmetry sequence $\overline{s}$ applied to a graph $G$ is in fact a sequence of $|V(G)|$ vertex symmetries $s_u, u\in V(G)$, as any composition of an odd number of symmetries is a symmetry. From now on, we will denote by $\overline{s}_u$ the symmetry applied to $u$ in a symmetry sequence $\overline{s}$.\\
Second, we define the equivalence relation induced by the rotations. Using this equivalence relation, we can define graphs in which vertices have an orientation rather than an ordering of their edges. 
}
Oriented simplicial complexes correspond to the equivalences classes of our labeled graphs:
\begin{definition}[Rotation Equivalence]
Two graphs $G$ and $H$ are rotation equivalent if and only if there exists a rotation sequence $\overline{r}$ such that $\overline{r} G=H$.
\end{definition}
{\em From now on and in the rest of this paper, we  will let $\Sigma=\varnothing$ in order to simplify notations---although all of the results of this paper carry through to graphs with internal states.}

\section{Graph distance vs geometrical distance}\label{sec:distances}

On the one hand in the world of simplicial complexes, two simplices are adjacent if they share a geometric point (a $0$--face). On the other hand in the world of graphs, two vertices are adjacent if they share an edge. These two notions do not coincide, as shown in Fig. \ref{fig:correspondance}. In order to understand the interplay between geometrical and graph distances, we first express the notion of $k$--face of a given simplex, in graph terms. We then provide graph-based condition that tell whether the $k$--face is shared by another simplex.




The way we express the notion of $k$--face in terms of graphs is as follows. Consider Figure \ref{fig:permutations}. Each port $p$ can be interpreted, geometrically, as the point opposite to where the gluing occurs. Then, a $k$-face $F$ can be described just by the set of ports--points that composes it.
\begin{definition}[Face]\label{def:face}
A $k$--face $F$ at vertex $u$ is a subset $\{p_0,\ldots,p_k\}$ of $k+1$ ports of a vertex $u$.
\end{definition}
Now, if a point $p$ belongs to a $k$--face at $u$, and a simplex $u'$ is glued on port $p$, this simplex no longer contains the $k$--face, as it excludes the point $u:p$. We can use this to characterize geometrically equivalent $k$--faces and hinges around them, i.e. paths along simplices that include them.

\begin{figure}[ht]
    \centering
    \includegraphics[scale=0.7]{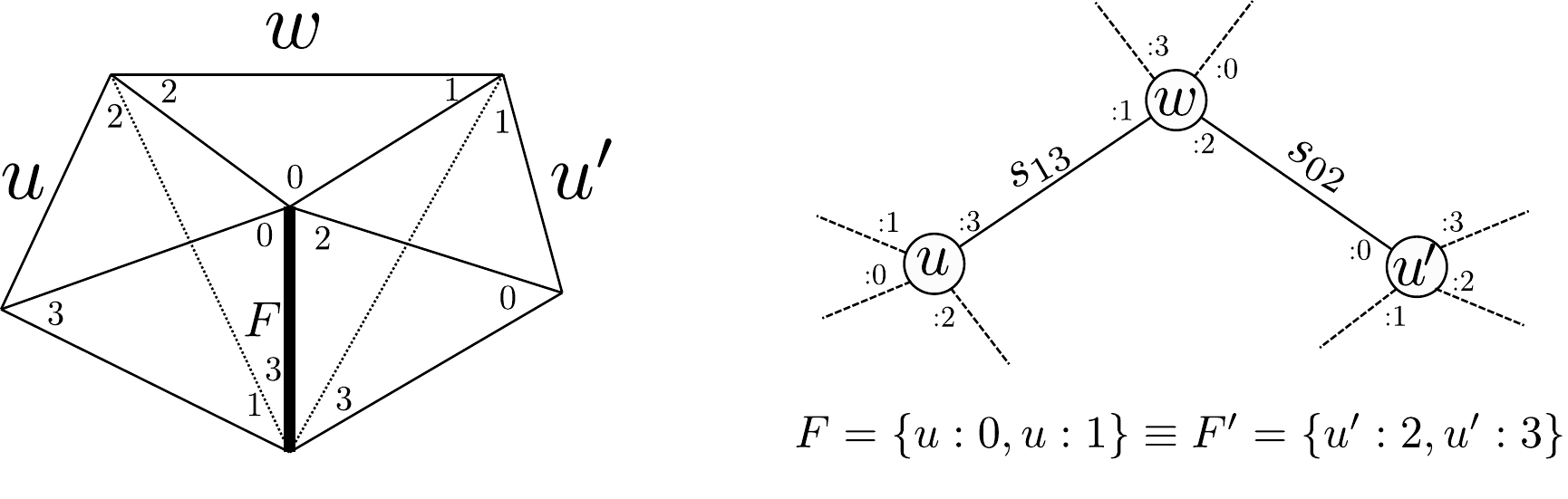}
    \caption{A hinge between $F$ at $u$ and $F'$ at $u'$.}
\end{figure}

\begin{definition}[Hinges between equivalent faces]\label{def:hinge}
Two $k$-faces $F$ at vertex $u$ and $F'$ at vertex $u'$ are said to be equivalent if and only if they are related by a {\em hinge}, i.e. if and only if there exists is a path $(u_i\ports p_i, \gamma_{i+1}, u_{i+1}\ports q_{i+1}) \in E(G)$ with $i = 0\ldots m$, $u_0=u$, $u_{m+1}=u'$, such that :
\begin{align}
p_{i},q_{i}&\notin \left(\prod_{j=1}^i \gamma_j\right) (F)\qquad\textrm{and}\qquad F'=\left(\prod_{j=1}^m \gamma_j\right) (F) \label{eq:hinge}
\end{align} 
where $p_{i}=p_0,\ldots, p_m$, whereas $q_{i}=q_1,\ldots, q_{m+1}$.
\end{definition}
When a gluing occurs on port $u:p$, it `covers' the points $u:\pi\setminus \{p\}$. Conversely, a $k$--face $F$ at $u$ is covered by all those gluings occurring at $\pi\setminus F$.
\begin{definition}[Border face]
Given a $k$-face $F$ at vertex $u$, consider every $F'$ at $u'$ that is equivalent to $F$. Its set of {\em covering semi-edges} is 
$$ S(G)\ \cap\ \bigcup_{u'}\ (u':\pi\setminus F').$$ 
If this set is non-empty, $F$ is a {\em border face}. 
\end{definition}
Sometimes a hinge can be closed-up into a cyclic hinge in a way that identifies a $k$--face, with a rotated/articulated version of itself, as in Fig. \ref{fig:torsion}. 
\begin{definition}[Torsion]
A {\em torsion} is a hinge around two distinct $k$-faces $F$ and $F'$ at $u$.
\end{definition}
Whilst such torsions may be useful in order to model certain kinds of parallel transport, we regard them as undesirable in this paper. Here is one tool to chase them out:
\begin{definition}[Normal form]
A path $\{u_i\ports p_i, \gamma_{i+1}, u_{i+1}\ports q_{i+1}\} \in E(G)$ with $i = 0\ldots m$ is in normal form if and only if for all $i$, $\gamma_{i+1} = s_{p_i q_{i+1}}$ and if the $n+1$ ports of $u_i$ are $\{p_i, r_1, \ldots, r_n \}$, then those of $u_{i+1}$ are $\{q_{i+1}, r_1, \ldots, r_n \}$.
\end{definition}
\begin{proposition} A cyclic hinge that can be put in normal form, is torsion-free. 
\end{proposition}
\begin{proof}
Say that the hinge has been put in normal form. Pick $F$ a $k$ face at $u_0$ such that the left equation \eqref{eq:hinge} is verified. In normal form $s_{p_i q_{i+1}}$ leaves $r_1, \ldots r_n$ unchanged. Obviously $p_0\notin F$, hence $s_{p_0 q_1}(F) = F$. And similarly for the next steps. Therefore the $F'$ of the right equation \eqref{eq:hinge} is $F$.
\end{proof}

Generally speaking, chasing out torsions is a difficult thing, because cyclic hinges may be arbitrary long. Unless we make further assumptions.
\begin{definition}[Star, Bounded-star]\label{def:bsgraphs}
Consider a graph $G$ and a vertex $u$ in $G$. A vertex $u'$ in $G$ is said to be a geometrical neighbor of $u$ if and only if they have an equivalent $k$--face. The {\em star} of $u$ is the subgraph induced by $u$ and its geometrical neighbors. It is denoted $\operatorname{Star}(G,u)$. A graph $G$ is said to be {\em bounded-star} of bound $s$ if and only if its hinges are of length less than or equal to $s$. 
\end{definition}

\begin{figure}[th]
\begin{center}
 \includegraphics{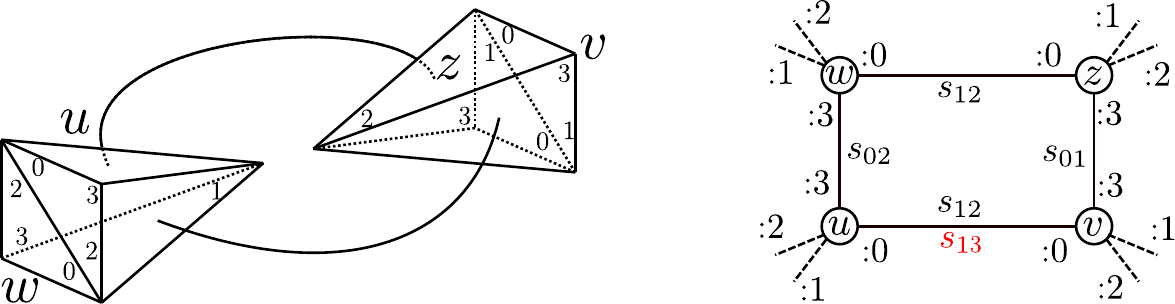}
 \end{center}
 \caption{\label{fig:torsion} With the black gluing between $u$ and $v$, the complex has the topology of a pinched ball, i.e. a doughnut whose hole has collapsed into a point. This constitutes an example of a pseudo-manifold (well-glued simplices) that is not a discrete manifold (the neighborhood of the pinch is not a ball). With the red gluing, the complex is torsioned. For instance, the $0$--faces $\{1\}$ and $\{2\}$ at $u$ are made equivalent in the sense of Def. \ref{def:hinge}.}
\end{figure}

\section{Causal Dynamics of Complexes}\label{sec:CDC}

We now recall the essential definitions of CGD, through their constructive presentation, namely as localizable dynamics. We will not detail, nor explain, nor motivate these definitions in order to avoid repetitions with \cite{ArrighiCGD,ArrighiIC,ArrighiCayleyNesme}. Still, notice that in \cite{ArrighiCGD,ArrighiIC,ArrighiCayleyNesme} this constructive presentation is shown equivalent to an axiomatic presentation of CGD, which establishes the full generality of this formalism. The bottom line is that these definitions capture all the graph evolutions which are such that information does not propagate information too fast and which act everywhere the same, see Fig. \ref{fig:CGD}.
\begin{figure}
\begin{center}
\includegraphics[scale=0.3]{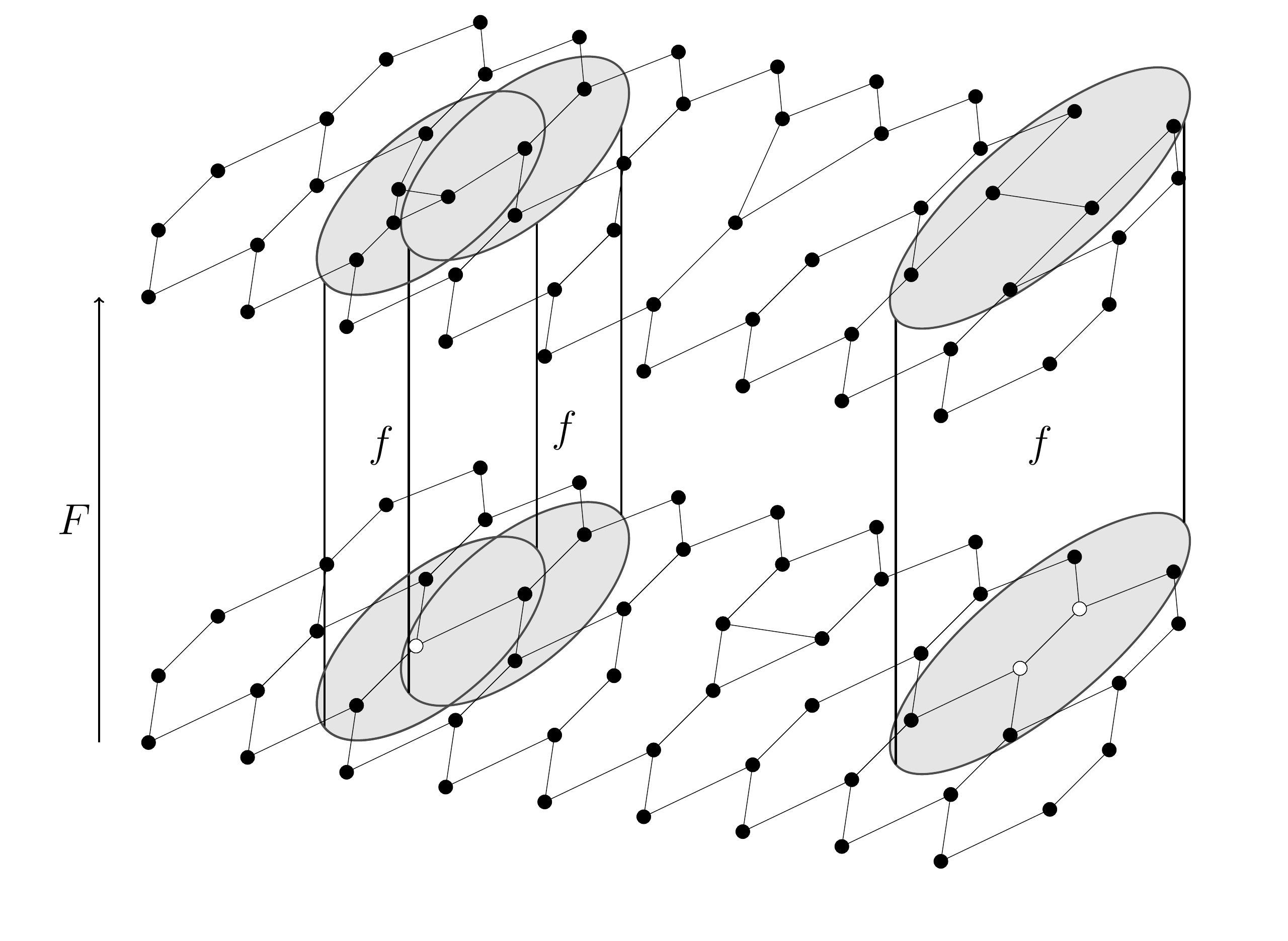}\vspace{-0.75cm}
\end{center}
\caption{{\em A Causal Graph Dynamics.} The whole graph evolves in a causal (information propagates at a bounded speed) and homogeneous (same causes lead to same effects) manner. This was proven equivalent to applying a local function $f$ to each subdisk of the input graph, producing small output graphs whose union make up the output graph.}\label{fig:CGD}
\end{figure}

\begin{definition}[Isomorphism]
An isomorphism is specified by a bijection $R$ from $V$ to $V$ and acts on a graph $G$ as follow:
\begin{itemize}
\item $V(R(G))=R(V(G))$
\item $(u:k, \gamma, v:l)\in E(G) \Leftrightarrow (R(u):k, R\circ \gamma\circ R^{-1}, R(v):l)\in E(R(G))$
\item $(u:k)\in S(G) \Leftrightarrow (R(u):k)\in S(R(G))$
\end{itemize}
Let $b$ be an integer number, and $\mathcal{F}(S)$ denote the finite subsets of a set $S$.
We similarly define the isomorphism $R^*$ specified by the isomorphism $R$ as the function acting on graphs $G$ such that $V(G)\subseteq \mathcal{F}(V.\{\varepsilon,1,...,b\})$, so that
$R^*(\{u.i,v.j,...\})=\{R(u).i,R(v).j,...\}$.
\end{definition}

\begin{definition} [Consistent]
Consider two graphs $G$ and $H$. Let $K=V(G)\cap V(H)$ and $L=V(G)\cup V(H)$. $G$ and $H$ are {\em consistent} if and only if 
for all $u:i$ in $K:\pi$, for all $v:j$ in $L:\pi$,
$$(u: i,\gamma, v: j) \in E(G)\vee (u:i) \in S(G)  \Longleftrightarrow (u: i,\gamma, v: j)\in E(H)\vee (u:i)\in S(H).$$
\end{definition}

\begin{definition}[Local Rule]
A function $f:\mathcal{D}^r \rightarrow \mathcal{G}$ is called a local rule if there exists some bound $b$ such that:
\begin{itemize}
\item For all disk $D$ and $v'\in V(f(D)) \Rightarrow v' \subseteq V(D).\{\varepsilon,1,...,b\}$.
\item For all graph $G$ and disks $D_1,D_2 \subset G$, $f(D_1)$ and $f(D_2)$ are consistent.
\item For all disk $D$ and isomorphism $R$, $f(R(D))=R^*(f(D))$\VL{.}\VS{, with $R^*(\{u.i,v.j,...\})=\{R(u).i,R(v).j,...\}$.}
\end{itemize}
\end{definition}
\VL{
\begin{definition}[Union]
The union $G\cup H$ of two consistent graphs $G$ and $H$ is defined as follow:
\begin{itemize}
\item $V(G\cup H)=V(G)\cup V(H)$
\item $E(G\cup H)=E(G)\cup E(H)$
\end{itemize}
\end{definition}
}
\begin{definition}[CGD]\cite{ArrighiCGD,ArrighiIC,ArrighiCayleyNesme}
A function $F$ from ${\cal G}$ to ${\cal G}$ is a {\em localizable dynamics, a.k.a Causal Graph Dynamics}, or CGD, if and only if there exists $r$ a radius and $f$ a local rule from ${\cal D}^r$ to ${\cal G}$ such that for every graph $G$ in ${\cal G}$, 
$$F(G)=\bigcup_{v\in G} f(G^r_v).$$
\end{definition}
To compute the image graph, a CGD could make use of the information carried out by the ports of the input graph. Thus, though the correspondence developed, they can readily be interpreted as ``Causal Dynamics of Colored Complexes''. If we are interested in ``Causal Dynamics of (Oriented) Complexes'' instead, we need to make sure that $F$ commutes with vertex-rotations. 
\begin{definition}[Rotation-commuting dynamics]
A CGD $F$ is rotation-commuting if and only if for all graph $G$ and all rotation sequence $\overline{r}$ there exists a rotation sequence $\overline{r}^*$ such that $ F(\overline{r} G)=\overline{r}^* F(G) $. Such an $\overline{r}^*$ is called a conjugate of $\overline{r}$.
\end{definition}
For local rules we will need a stronger version of this:
\begin{definition}[Strongly-rotation-commuting local rule]
A local rule $f$ is strongly-rotation-commuting if and only if for all intersecting pairs of disks $G = D_1 \cup D_2$ and for all rotation sequence $\overline{r}$, the conjugate rotation sequences $\overline{r}^*_1$ and $\overline{r}^*_2$ defined through $\overline{r}^*_i f(D_i) = f(\overline{r}D_i)$, $i=1,2$ coincide on $f(D_1) \cap f(D_2)$.
\end{definition}
When is a CGD rotation-commuting? Can we decide, given the local rule $f$ of a CGD $F$, whether $F$ is rotation-commuting? The difficulty is that being rotation-commuting is a property of the global function $F$. Indeed, a first guess would be that $F$ is rotation-commuting if and only if $f$ is rotation-commuting, but this turns out to be false. 
\begin{example}[Identity] Consider the local rule of radius $1$ over graphs of degree $2$ which acts as the identity in every cases but those given in Fig. \ref{fig:noncom}. Because of these two cases, the local rule makes use the information carried out by the ports around the center of the neighborhood. It is not rotation-commuting. Yet, the CGD it induces is just the identity, which is trivially rotation-commuting.
\begin{figure}
\begin{center}
\includegraphics[scale=0.6]{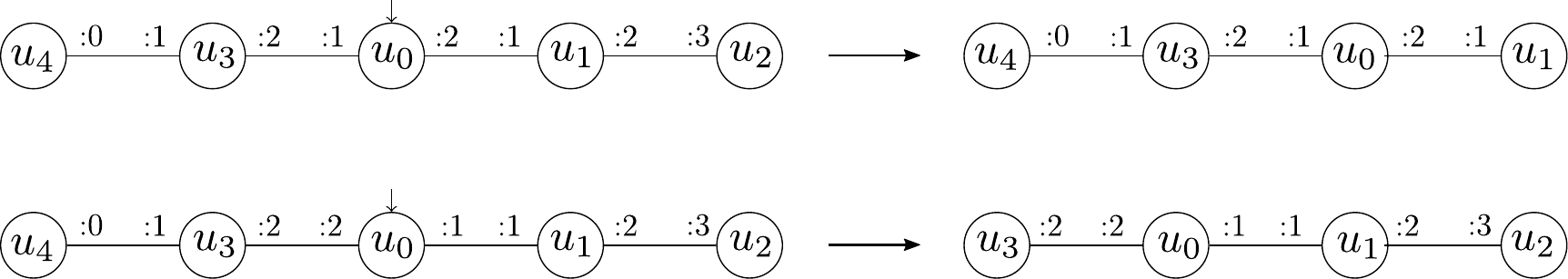}
\end{center}
\caption{A non-rotation commuting local rule which induces a rotation commuting CGD.}\label{fig:noncom}
\end{figure}
\end{example}
Thus, unfortunately, rotation-commuting $F$ can be induced by non-rotation-commuting $f$.
Still, there always exists a strongly-rotation-commuting $f$ that induces $F$.
\begin{proposition}
Let $F$ be a CGD. $F$ is rotation-commuting if and only if there exists a strongly-rotation-commuting local rule $f$ which induces $F$.
\end{proposition}
\begin{proof}(Outline). $[\Leftarrow]$ Trivial. $[\Rightarrow]$ Given a rotation commuting CGD $F$ induced by some local rule $f$ that is not necessarily strongly rotation commuting itself, we construct a local rule $\tilde{f}$ that is a strongly-rotation-commuting and still induces $F$. The construction uses the fact that $F$ is rotation commuting to force $\tilde{f}$ to adopt an homogeneous behavior over the sets of disks of the form $\{\overline{r}D~~|~~\overline{r}\ \textrm{ a rotation sequence}\}$ (i.e rotation equivalent copies of the same disk). The complete proof is in Appendix \ref{app:pfrotationcommuting}.
\end{proof}\\
The point of this proposition is that having made this global property, local, makes it decidable. 
\begin{proposition}[Decidability of rotation commutation]
Given a local rule $f$, it is decidable whether $f$ is strongly-rotation-commuting.
\end{proposition}
\begin{proof} There exists a simple algorithm to verify that $f$ is strongly-rotation-commuting.
Let $r$ be the radius of $f$.
We can check that for all disk $D\in\mathcal{D}^r$ and for all vertex rotation $r_u$, $u\in V(D)$, we have the existence of a rotation sequence $\overline{r}$ such that $f(r_u D)=\overline{r}^*f(D)$.\\ As the graph $f(D)$ is finite, there is finite number of rotation sequences $\overline{r}^*$ to test. Notice that as $f$ is a local rule, changing the names of the vertices in $D$ will not change the structure of $f(D)$ and thus we only have to test the commutation property on a finite set of disks.
\end{proof}
\begin{definition}[CDC]\label{def:CDC}
A {\em Causal Dynamics of Complexes} is a rotation-commuting CGD.
\end{definition}

\section{Pachner Moves}\label{sec:pachner}

\noindent {\em Bistellar moves.} Given a tetrahedron $\Delta_3$, there is a canonical way to obtain is border, $\partial \Delta_3$, as four glued triangles. In terms of graphs, given a single vertex of degree $4$, there is a canonical way to obtain a graph made of four vertices of degree $3$ that represents its border. This works as follows: 1. Interpret the vertex as a colored tetrahedron; so that each point has a color; 2. Reinterpret each facet as a vertex, and each gluing along a segment, as an edge between the ports of colors that of the points opposite the segment. This is the way we obtain:

\begin{definition}[$\partial \Delta_{n+1}$]\label{def:sphere}
We call the canonical sphere of dimension $n$, and denote $\partial \Delta_{n+1}$, the complete graph of size $n+2$ having vertices $v_0,...,v_{n+1}$ and edges of the form $(v_i\ports j,s_{ij},v_j\ports i)$ for $i\neq j$ and $i,j \in \{0,...,n+1\}$.
\end{definition}
{\bf  Soundness.} All hinges are in normal form hence not torsioned. \hfill $\Box$

A triangle $H$ can always be viewed as being a subcomplex of the boundary of a tetrahedron. Its complement with respect to the tetrahedron yields three other triangles $H^*$. More generally and in terms of graphs, whenever $H$ is a subgraph of $\partial \Delta_{n+1}$, we can construct its complement $H^*$ with respect to $\partial \Delta_{n+1}$. 

When we have a triangle $H$ lying inside a larger complex $G$, we can decide to replace $H$ by $H^*$ in $G$. This amounts to subdividing it into three, see Fig. \ref{fig:bist}. More generally and in terms of graphs, whenever $H$ is an induced subgraph of $\partial \Delta_{n+1}$ and lies inside a larger graph $G$, we can decide to replace $H$ by $H^*$ in $G$. A bistellar move does exactly that: it replaces a piece a sphere by its complement, it is intuitive therefore that it is a homeomorphism:
\begin{figure}[ht]
    \begin{center}
    \includegraphics[scale=0.7]{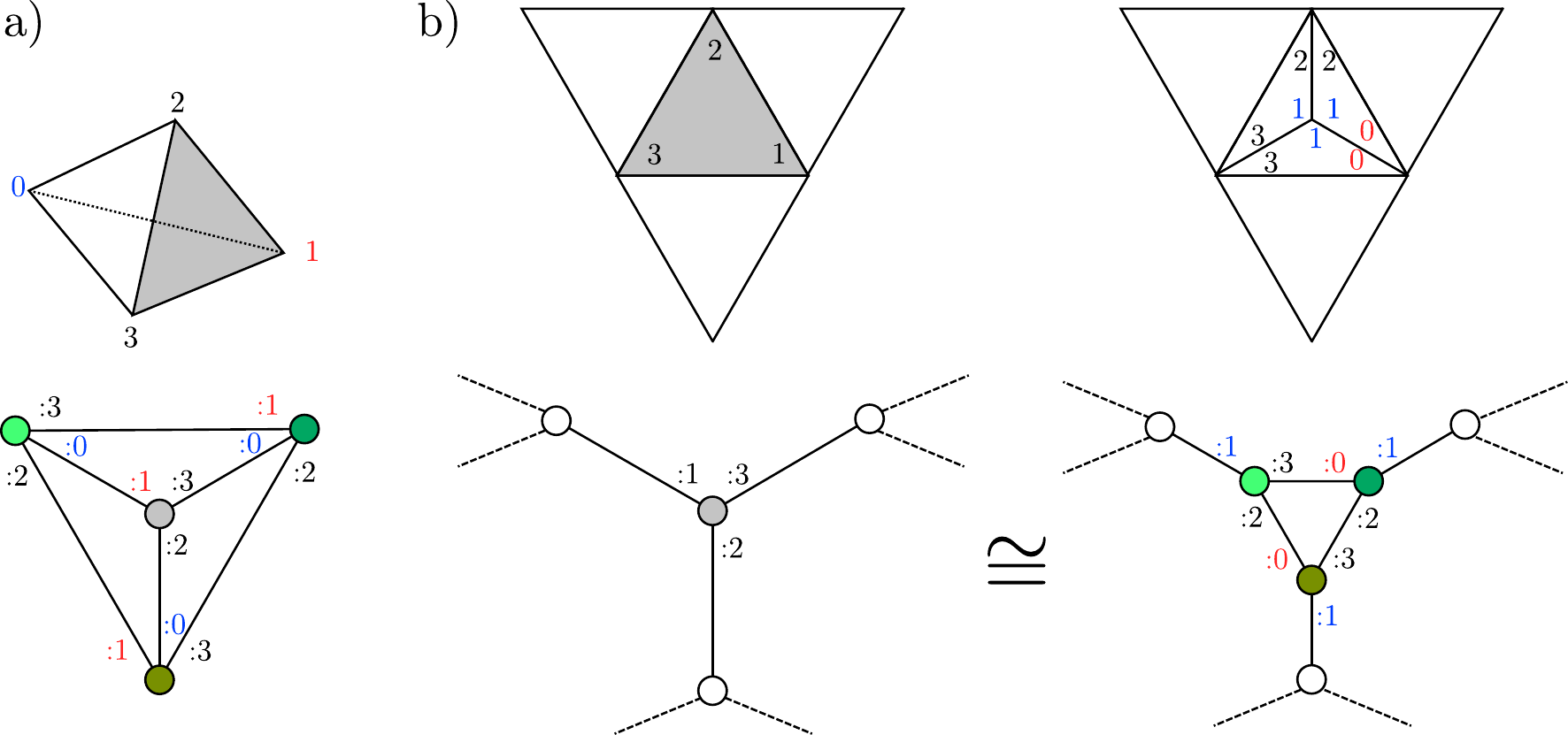}
    \end{center}\label{fig:bist}
    \caption{$a)$ The canonical sphere of dimension $2$. $b)$ The bistellar move obtained by taking as the subgraph $H$, the single vertex in gray.}
\end{figure}
\begin{definition}[Bistellar move $G.H$]
Let $G$ be a graph and $H$ be a subgraph of $G$ such that $H$ is a strict subset of a $\partial \Delta_{n+1}$, and $(G\setminus H) \cap \partial\Delta_{n+1} = \emptyset$. Let us call $\overline{s}$ the symmetry sequence $(s_{01})_{u\in H^*}$. The graph $G.H$ is the graph where $H$ has been replaced by $\overline{s}H^*$ as follows.
First, add $\overline{s}H^*$ to the graph. Second, for each edge $e = (v\ports p, \gamma ,u\ports q)$ between a vertex $v$ of $H$ and a vertex $u$ of $G\setminus H$, notice there is a unique edge $e'=(v'\ports p', \gamma'  , v\ports p ) \in E(\overline{s}\partial\Delta_{n+1})$, and replace both $e$ and $e'$ by the edge $(v'\ports p', \gamma \circ \gamma' , u\ports q)$. Similarly, for every semi-edge $e=(v:p)$ of $H$, notice there is a unique edge $e'=(v'\ports p', \gamma'  , v\ports p ) \in E(\overline{s}\partial\Delta_{n+1})$, and replace both $e$ and $e'$ by the semi-edge $(v'\ports p')$. Third, remove the vertices of $H$.
\end{definition}
Notice that $\overline{s}$ flips the orientation of $H^*$ relative to $H$ in order to match the orientation of $G$, and that the use of $s_{01}$ for this purpose is without loss of generality, as one should also allow for vertex rotations.

\noindent {\em Shellings.}  Given a $2$--dimensional complex having a triangle with two of its sides on the boundary, we can decide to grow the complex by gluing two sides of a new triangle there. Similarly and in terms of graphs, whenever a vertex of degree $n+1$ has $k$ free ports, we can decide to connect them with $k$ ports of a new, otherwise unconnected vertex. A graph-local inverse shelling indeed consists in adding a new vertex to a graph by connecting it to a vertex having free ports:
\begin{definition}[Graph-local (inverse) shellings]\label{def:shelling}
Let $G$ be a graph and $u$ a vertex of $G$. Let $S$ be a subset of at most $n$ free ports of $u$, i.e. such that $(u:p)\in S(G)$ for all $p\in S$. The graph $G.S$ is the graph where a fresh vertex $v$ has been added, as well as edges  $(u:p, s_{01} , v:s_{01}(p))$. We say that $G.S$ is an {\em graph-local inverse shelling} of $G$, and conversely that $G$ is a {\em graph-local shelling} of $G.S$.
\end{definition}
There is difference, however, between this graph-local notion of shelling and the standard notion of shelling upon complexes. Indeed, as was pointed out in Section \ref{sec:complexesasgraphs}, in a $2$--dimensional complex two boundary segments may be consecutive without this locality being apparent in the corresponding graph. Standard inverse shellings are definitely more general, as they allow gluing a fresh triangle there. Phrased in terms of graphs, they translate into:
\begin{definition}[Standard (inverse) shellings]\label{def:stdshelling}
Let $G$ be a graph, $u$ be a vertex of $G$, and $F$ be a border $k$-face at $u$, having exactly $n-k$ covering semi-edges $(u_i:p_i)$. The graph $G.F$ is the graph where a fresh vertex $v$ has been added, and each semi-edge $(u_i:p_i)$ has been replaced by an edge  $(u_i:p_i, s_{01} , v:s_{01}(p_i))$, without creating any torsion. We say that $G.F$ is an {\em standard inverse shelling} of $G$, and conversely that $G$ is a {\em standard shelling} of $G.F$.
\end{definition}
As an example of this definition, consider filling, with a new tetrahedron $v$, the hole in the ball at the top of Fig. \ref{fig:shellings} a), as in the top of Fig. \ref{fig:shellings} c). The $0$--face $F$ stands for the geometrical point that will be covered by $v$. As the three covering semi-edges of $F$ will be replaced by edges, $F$ will no longer be a border face. Indeed, although $v$ introduces a new semi-edge, that one is not a covering semi-edge of $F$.\\

\begin{figure}[ht]
	\begin{center}
    \includegraphics[scale=0.7]{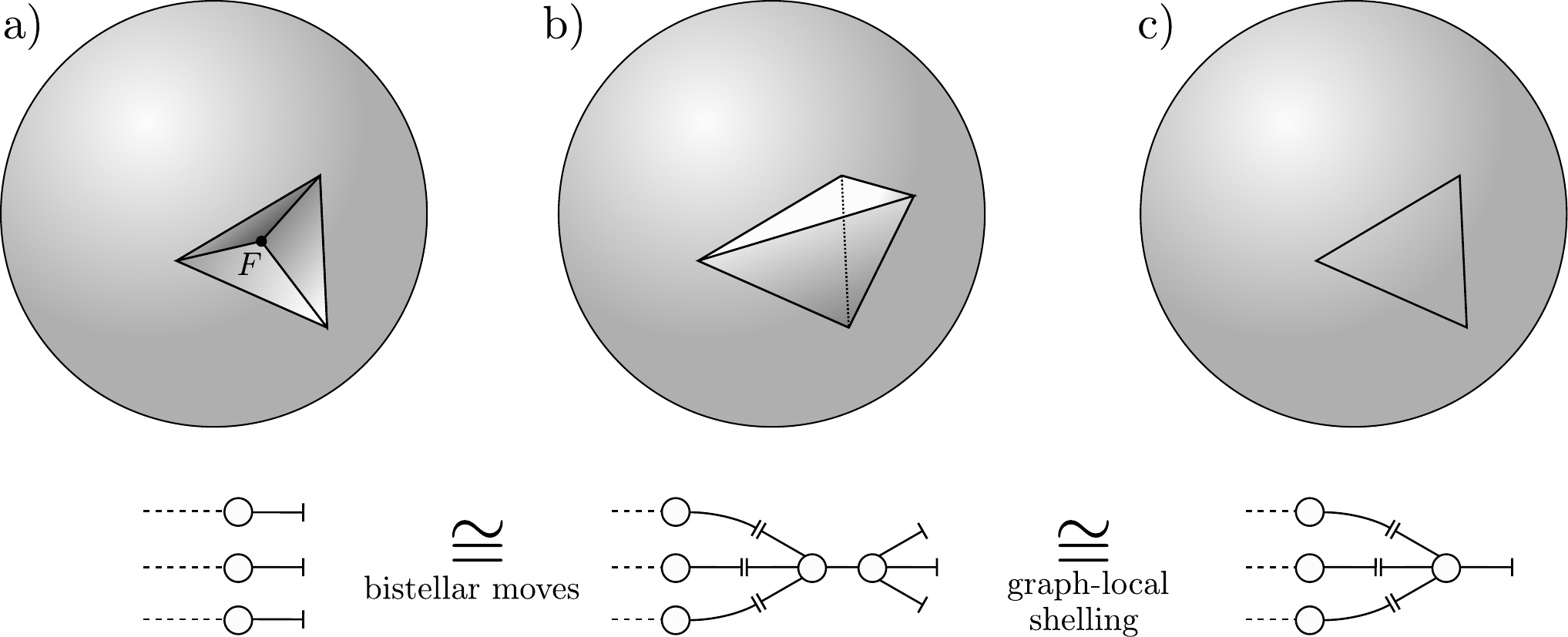}
    \end{center}
    \caption{The standard inverse shelling obtained via bistellar moves and a graph-local shelling.}\label{fig:shellings}
\end{figure}

Fortunately, standard (inverse) shellings can always be recovered from a succession of rotations, Bistellar moves, graph-local (inverse) shellings: 
\begin{definition}[Graph-local Pachner moves]\label{def:pachner}
We call {\em graph-local Pachner moves} the union of vertex rotations, bistellar moves and graph-local (inverse) shellings.
\end{definition}
\begin{proposition}[Recovering standard shellings]\label{prop:stdshelling}
Standard (inverse) shellings are compositions of graph-local Pachner moves. 
\end{proposition}
\begin{proof}
Consider a graph $G$ with a border $k$-face $F$ having exactly $n-k$ covering semi-edges as in Fig. \ref{fig:shellings} a). 
We want to perform the standard inverse shelling $G.F$, adding a fresh vertex $v$, using only graph-local Pachner moves. As an intermediate step, consider $G'$ the graph $G.F.S$ where $S$ is the set of semi-edges of $v$, as in Fig. \ref{fig:shellings} b).
The graphs $G$ and $G'$ are homeomorphic and have the same border, therefore they are related by a sequence of bistellar moves, as was shown by \cite{Casali}. Finally, by a graph-local shelling we obtain $G.F$ as in \ref{fig:shellings} c). 
\end{proof}

In the setting of simplicial complexes, Pachner moves \cite{Pachner,Lickorish} are well-known to generate all the homeomorphisms between combinatorial manifolds, and only the homeomorphisms. As a corollary of the above proposition the same holds true for graph-local Pachner moves: 
\begin{definition}[Discrete manifold]
A graph $G$ of degree $|\pi|=n+1$ is discrete manifold if and only if for each vertex $u\in V(G)$, there exists a sequence of graph-local Pachner moves sending $Star(G,u)$ onto $\Delta_n$.
\end{definition}
\begin{corollary}[Homeomorphism]\label{th:homeo} Consider $M$ and $M'$ two piecewise-linear manifolds, and let $G$ and $G'$ be the discrete manifolds obtained as their respective triangulations into simplicial complexes. $M$ and $M'$ are piecewise-linearly homeomorphic if and only if $G$ and $G'$ are related by a sequence of graph-local Pachner moves.
\end{corollary}
Notice that, albeit expensive computationally, it is decidable whether a $n<4$--dimensional complex is homeomorphic to $\Delta_n$, see \cite{Lazarus} and \cite{Kuperberg} (Prop. 3.1). Homeomorphism in general becomes undecidable for $n\geq 4$ \cite{Lazarus}. Notice also that discrete manifolds are not always simplicial complexes. For instance the self-glued triangle is not a simplicial complex, as points of the same simplex get identified. In dimension $n\leq 2$, this remark seems innocuous, as any discrete manifold is related, via Pachner moves, to a simplicial complex. In dimension $n=3$, we conjecture that this is still the case if and only if each simplex has no more that two identified points, and that one extra move suffices to make this true in all cases.

\section{Causal Dynamics of Discrete Manifolds}\label{sec:CDDM}

The results in this section crucially rely on 
\begin{lemma}[Past subgraph]\cite{ArrighiIC}\label{lem:pastsubgraph}
Consider $F$ a CGD induced by the local rule $f$ of radius $r$ (i.e. diameter $d=2r+1$). 
Consider a graph $G$, a vertex $v$ in $G$, a vertex $v'$ in $f(G^r_v)$, and a disk ${F(G)}^{r'}_{v'}$ (i.e. of diameter $d'=2r'+1$).
Then this disk is a subgraph of $F(G^{2rr'+r+r'}_v)$. Notice that the disk $G^{2rr'+r+r'}_v$ has diameter $d''=d'd$.
\end{lemma}

\noindent  {\em Bounded-star preserving.} We will now restrict to CGD so that they preserve the property of a graph being bounded-star. Indeed, we have seen that graph distance between two vertices does not always correspond to the geometrical distance between the two triangles that they represent. With CDC, we were guaranteeing that information does not propagate too fast with respect to the graph distance, but not with respect to the geometrical distance. The fact that the geometrical distance is less than or equal to the graph distance is falsely reassuring: the discrepancy can still lead to unwanted phenomenon as depicted in Fig. \ref{fig:boundedstar}.
\begin{figure}\begin{center}
\includegraphics[width=\textwidth]{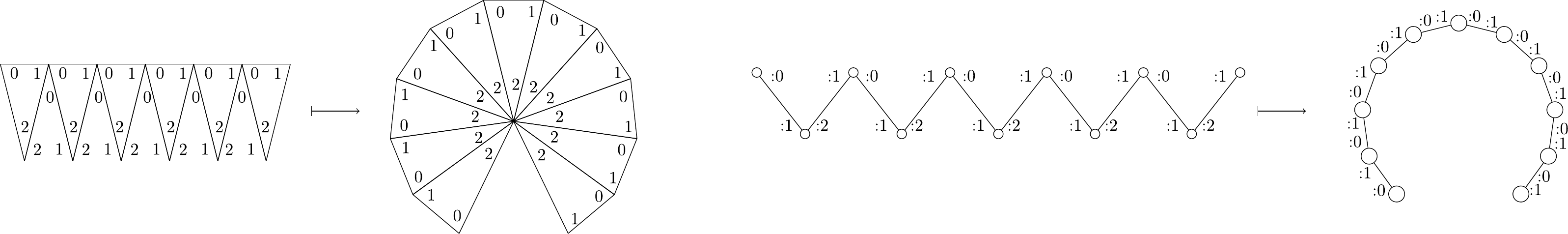}
\caption{An unwanted evolution: sudden collapse in geometrical distance. \label{fig:boundedstar}
{\em Left:} in terms of complexes. {\em Right:} In terms of graph representation.}
\end{center}\end{figure}
Of course we may choose not to care about geometrical distance. But if we do care, then we must not let that happen. One solution is to make the graphs are $s$--bounded-star. This will relate the geometrical distance and the graph distance by a factor $s$. As a consequence, the guarantee that information does not propagate too fast with respect to graph distance will induce its counterpart in geometrical distance. This will forbid the sudden collapse in geometrical distance of Fig. \ref{fig:boundedstar}. More generally it will enforce a bounded-density of information principle. Of course, we must then ensure that the CGD we use preserve $s$--bounded-star graphs:
\begin{definition}[Bounded-star preserving]
A CGD $F$ is \emph{bounded-star preserving with bound $s$} if and only if for all $s$--bounded-star graph $G$, $F(G)$ is also $s$--bounded-star.
\end{definition}
\begin{proposition}\label{prop:boundedstar}
Consider $F$ a CGD induced by a local rule $f$ of radius $r$. 
$F$ is bounded-star preserving with bound $s=2r'$ if and only if for any $s$--bounded-star $D$ in ${\cal D}^{2rr'+r+r'}$, $F(D)$ is also $s$--bounded-star. Therefore, given a local rule $f$, it is decidable whether its induced $F$ is bounded-star preserving. 
\end{proposition}
\begin{proof}$[\Rightarrow]$ Trivial.
$[\Leftarrow]$ By contradiction suppose that there is an $s$--bounded-star graph $G$ such that $F(G)$ has an hinge $h$ of size $s'=2r'+1$, and yet that all $s$--bounded-star disks of radius ${2rr'+r+r'}$ are mapped into $s$--bounded-star graphs. Next, take $v'$ in the middle of $h$, and $v$ in $G$ such that $v'$ in $f(G^r_v)$. By Lemma \ref{lem:pastsubgraph}, $h$ appears in $F(G_v^{2rr'+r+r'})$, which contradicts our hypothesis.    
\end{proof}

\noindent  {\em Torsion-free preserving.} Second, amongst bounded-star preserving CGD, we will restrict to those that preserve the property of not having torsion. 
\begin{definition}[Torsion-free preserving]\label{prop:notorsion}
An $s$--bounded-star preserving CGD $F$ is \emph{torsion-free preserving} if and only if for all $s$--bounded-star graph $G$ without torsion, $F(G)$ is without torsion.
\end{definition}
\begin{proposition}
Consider $F$ an $s$--bounded-star CGD induced by a local rule $f$ of radius $r$. 
$F$ is torsion-free preserving if and only if for any $s$--bounded-star $D$ in ${\cal D}^{2rr'+r+r'}$ without torsion, $F(D)$ is also without torsion. Therefore, given a local rule $f$, it is decidable whether its induced $F$ is torsion-free preserving. 
\end{proposition}
\begin{proof} As for Prop. \ref{prop:boundedstar}. \end{proof}

\noindent  {\em Discrete-manifold preserving.} Third, amongst torsion-free bounded-star preserving CGD, we will restrict to those that preserve the property of being a discrete manifold. 
\begin{definition}[Discrete-manifold preserving]
An torsion-free $s$--bounded-star preserving CGD $F$ is \emph{discrete-manifold preserving} if and only if for all $s$--bounded-star discrete manifold $G$, then $F(G)$ is a discrete manifold.
\end{definition}
\begin{proposition}
Consider $F$ a torsion-free $s$--bounded-star preserving CGD induced by a local rule $f$ of radius $r$. 
$F$ is discrete-manifold preserving if and only if for any $s$--bounded-star discrete manifold $D$ in ${\cal D}^{2rr'+r+r'}$, $F(D)$ is also a discrete manifold. Therefore, in dimension $n\leq 3$, given a local rule $f$, it is decidable whether its induced $F$ is discrete-manifold preserving. 
\end{proposition}
\begin{proof}
As for Prop. \ref{prop:boundedstar}. Checking whether $F(G_v^{2rr'+r+r'})$ is a discrete-manifold is indeed possible in dimension $n\leq 3$, cf. \cite{Lazarus} and \cite{Kuperberg} (Prop. 3.1).
\end{proof}

\begin{definition}[CDDM]\label{def:CDDM}
A {\em Causal Dynamics of Discrete Manifolds} is a torsion-free $s$--bounded-star discrete-manifold preserving CGD.
\end{definition}

\section{Conclusion}\label{sec:conclusion}

{\em Results in context.} In \cite{ArrighiCGD,ArrighiIC,ArrighiCayleyNesme} two of the authors, together with Dowek and Nesme, generalized Cellular Automata theory to arbitrary, time-varying graphs. I.e. they formalized the intuitive idea of a labeled graph which evolves in time, subject to two natural constraints: the evolution does not propagate information too fast; and it acts everywhere the same. Some fundamental facts of Cellular Automata theory were shown to carry through, for instance that these Causal Graph Dynamics (CGD) admit a characterization as continuous functions and that their inverses are also CGD. 

The motivation for developing these CGD was to ``free Cellular Automata off the grid'', so as to be able to model any situation where agents interact with their neighbors synchronously, leading to a global dynamics in which the states of the agents can change, but also their topology, i.e. the notion of who is next to whom. A first motivating example was that of a mobile phone network. A second example was that of particles lying on a surface and interacting with one another, but whose distribution influences the topology the surface (cf. Heat diffusion in a dilating material, discretized General Relativity \cite{Sorkin}). However, CGD seemed quite appropriate for modeling the first situation (or at least a stochastic version of it), but not the second. Indeed, having freed Cellular Automata off the grid, one could no longer interpret arbitrary graphs as surface, in general. 

The present paper solves this problem by proposing a rigorous definition of ``Causal Dynamics of Complexes'' (CDC) and ``Causal Dynamics of Discrete Manifolds'' (CDDM). Essentially this shows that CGD can be ``tied up again to complexes and even to discrete manifolds'', at the cost of additional restrictions: rotation-commutation (CDC), bounded-star preservation, torsion-free preservation, discrete-manifold preservation (CDDM). The first restriction allows us to freely rotate simplices. The second restriction allows us to map geometrical distances into graph distances. The third restriction makes sure that no torsion gets introduced. The fourth restriction makes sure that the neighborhood of every point remains a ball. The first and second are decidable independently. Imposing the second makes the third decidable, and fourth, but in dimensions $n<4$ only. An earlier version investigated the $2$--dimensional case in order to gain intuitions \cite{ArrighiSURFACES}. This paper provides its non-trivial generalization to $n$-dimensions: the third and fourth conditions, for instance, were vacuous in the $2$--dimensional case. In order to tackle it, we translated the notion of manifold homeomorphism the vocabulary of labeled graphs.

Notice that, since these CDC and CDDM are a specialization of CGD by construction, several theoretical results about them follow as mere corollary from \cite{ArrighiCGD,ArrighiIC,ArrighiCayleyNesme}---that we have not mentioned. For instance, CDC/CDDM of radius $1$ are universal, composable, characterized as the set of continuous functions from complexes to complexes with respect to the Gromov-Hausdorff-Cantor metric upon isomorphism classes. These results deserve to be made more explicit, but they already are indicators of the generality of the model.

{\em Comparison with Crystallizations/Gems.} This paper conducted a thorough comparison between discrete geometries and graphs, by investigating the natural encoding of complexes into their dual graphs. This encoding was made precise. The discrepancy between geometrical distance and graph distance was analyzed. The notion of manifold was characterized, through a graph-local version of Pachner moves. Another, very well-developed correspondence between simplicial complexes and labeled graphs goes under the name of `crystallizations' \cite{Ferri}. Phrased in the vocabulary of Def. \ref{def:graphs}, this means restricting to bipartite graphs (i.e. w.r.t. to labels in $\Sigma=\{0,1\}$, say) that are edge--colored (i.e. edges are between equal ports) and have, as gluings, the identity. Because this gluing is an even permutation, $0$--labeled vertices are oriented one-way, and $1$--labeled vertices are oriented the other way. These constraints may seem cumbersome at first; for instance constructing a sphere becomes much more involved than Def. \ref{def:sphere}. Yet, a closer look shows a key advantage: by construction, crystallizations do not have torsion. The subset of crystallizations that represent discrete manifolds is usually referred to as `gems' (i.e. graph-encoded manifolds). Homeomorphism between gems can again be captured by moves. Traditionally the moves that have been studied are the so-called `dipole moves', but unfortunately these are not graph-local (a global condition needs be checked prior to application). Lately, however, \cite{Izmestiev} developed an equivalent of Bistellar moves, called `cross-flips moves', which captures homeomorphism between closed discrete manifolds, in a graph-local way. This has been extended to discrete manifolds with borders in \cite{Juhnke}---but the (inverse) shellings are again not graph-local. Yet, \cite{Juhnke} also contains the gems--version the result by \cite{Casali} that allowed us to prove that graph-local (inverse) shellings are enough. Thus, all the results of this paper can readily be ported to crystallizations/gems. Still, there will be a price to pay: the number of cross-flip moves is in $O(2^n)$ \cite{Juhnke}, whereas bistellar moves grow as $O(n)$.

\section*{Acknowledgments}
This work has been funded by the ANR-12-BS02-007-01 TARMAC grant and the STICAmSud project 16STIC05 FoQCoSS. The authors acknowledge enlightening discussions Gilles Dowek, Ivan Izmestiev, Pascal Lienhardt, Luca Lionni, Christian Mercat, Michele Mulazzani, Zizhu Wang. 

\bibliographystyle{plainurl}
\bibliography{biblio}

\begin{thebibliography}{10}

\bibitem{LollCDT}
Jan Ambj{\o}rn, Jerzy Jurkiewicz, and Renate Loll.
\newblock Emergence of a 4d world from causal quantum gravity.
\newblock {\em Physical Review Letters}, 93(13):131301, 2004.

\bibitem{ArrighiCGD}
P.~Arrighi and G.~Dowek.
\newblock {Causal graph dynamics}.
\newblock In {\em Proceedings of ICALP 2012, Warwick, July 2012, LNCS}, volume
  7392, pages 54--66, 2012.

\bibitem{ArrighiCayleyNesme}
P.~Arrighi, S.~Martiel, and V.~Nesme.
\newblock {Cellular automata over generalized Cayley graphs}.
\newblock {\em Mathematical Structures in Computer Science. Pre-print
  arXiv:1212.0027}, 18:340--383, 2018.

\bibitem{ArrighiSURFACES}
P.~Arrighi, S.~Martiel, and Z.~Wang.
\newblock {Causal dynamics over Discrete Surfaces}.
\newblock In {\em Proceedings of DCM'13, Buenos Aires, August 2013, EPTCS},
  2013.

\bibitem{ArrighiIC}
Pablo Arrighi and Gilles Dowek.
\newblock Causal graph dynamics (long version).
\newblock {\em Information and Computation}, 223:78--93, 2013.

\bibitem{Casali}
Maria~Rita Casali.
\newblock A note about bistellar operations on pl-manifolds with boundary.
\newblock {\em Geometriae Dedicata}, 56(3):257--262, 1995.

\bibitem{Ferri}
Massimo Ferri, Carlo Gagliardi, and Luigi Grasselli.
\newblock A graph-theoretical representation of pl-manifolds—a survey on
  crystallizations.
\newblock {\em Aequationes Mathematicae}, 31(1):121--141, 1986.

\bibitem{GiavittoMGS}
J.L. Giavitto and A.~Spicher.
\newblock Topological rewriting and the geometrization of programming.
\newblock {\em Physica D: Nonlinear Phenomena}, 237(9):1302--1314, 2008.

\bibitem{Izmestiev}
Ivan Izmestiev, Steven Klee, and Isabella Novik.
\newblock Simplicial moves on balanced complexes.
\newblock {\em Advances in Mathematics}, 320:82--114, 2017.

\bibitem{Juhnke}
Martina Juhnke-Kubitzke and Lorenzo Venturello.
\newblock Balanced shellings and moves on balanced manifolds.
\newblock {\em arXiv preprint arXiv:1804.06270}, 2018.

\bibitem{MeyerLove}
A.~Klales, D.~Cianci, Z.~Needell, D.~A. Meyer, and P.~J. Love.
\newblock Lattice gas simulations of dynamical geometry in two dimensions.
\newblock {\em Phys. Rev. E.}, 82(4):046705, Oct 2010.
\newblock \href {http://dx.doi.org/10.1103/PhysRevE.82.046705}
  {\path{doi:10.1103/PhysRevE.82.046705}}.

\bibitem{Kuperberg}
Greg Kuperberg.
\newblock Algorithmic homeomorphism of 3-manifolds as a corollary of
  geometrization.
\newblock {\em arXiv preprint arXiv:1508.06720}, 2015.

\bibitem{Lazarus}
Francis Lazarus and Arnaud de~Mesmay.
\newblock Undecidability in topology.
\newblock 2017.

\bibitem{Lickorish}
WB~Raymond Lickorish.
\newblock Simplicial moves on complexes and manifolds.
\newblock {\em Geometry and Topology Monographs}, 2(299-320):314, 1999.

\bibitem{Pachner}
Udo Pachner.
\newblock Pl homeomorphic manifolds are equivalent by elementary shellings.
\newblock {\em European journal of Combinatorics}, 12(2):129--145, 1991.

\bibitem{Sorkin}
R.~Sorkin.
\newblock {Time-evolution problem in Regge calculus}.
\newblock {\em Phys. Rev. D.}, 12(2):385--396, 1975.

\end{thebibliography}

\appendix

\section{Proofs of locality of rotation-commutation}\label{app:pfrotationcommuting}

\begin{lemma} \label{lem:rot_comm}
Let $G$ be a graph and $r_u, r'_v$ two vertex rotations, with $u, v$ two distinct vertices of $V(G)$. We have 
$(r_u \circ r'_v) G = (r'_v \circ r_u) G$. 
\end{lemma}
\begin{proof}
We just check the equality above, for each vertex and edge. This is trivial except for those edges of the form $(u:p, \gamma, v:q)$.
\begin{description}
\item[Case $(r_u\circ r'_v)$.] The gluing $\gamma$ goes into $r'\circ \gamma$ and then into $r^\prime\circ \gamma\circ r^{-1}$.
\item[Case $(r'_v\circ r_u)$.] The gluing $\gamma$ goes into $\gamma\circ r^{-1}$ of $u:r(p)$ and then into $r'\circ \gamma\circ r^{-1}$.
\end{description}
\end{proof}
As a consequence of Lemma \ref{lem:rot_comm} it makes sense to define $(\overline{r}_i / u)$ as the ordered gathering of the vertex rotations at $u$ in $\overline{r}_i$. Then, two rotation sequences $\overline{r}_1, \overline{r}_2$ can be said to be consistent with one another if, whenever they act upon a common vertex $u$, we have $(\overline{r}_1 / u)=(\overline{r}_2 / u)$.
\begin{lemma}\label{lem:rot_merge}
For all finite set of graphs $G_1,...,G_n$ and for all consistent set of rotation sequences $\overline{r}_1,...,\overline{r}_n$, if $G_1,...,G_n$ are consistent with each other, and $\overline{r}_1 G_1,...,\overline{r}_n G_n$ are consistent with each other, then there exists a rotation sequence $\overline{r}$ such that:
$$ \bigcup_{i\in\{1,...,n\}} \overline{r}_i G_i =\overline{r} \bigcup_{i\in\{1,...,n\}} G_i $$
\end{lemma}
\begin{proof}
Notice that we only need to prove this result for the union of two graphs.
Let us consider two consistent graphs $G_1,G_2$ and two consistent rotation sequences $\overline{r}_1,\overline{r}_2$ such that $\overline{r}_1 G_1,\overline{r}_2 G_2$ are consistent. We are now going to show that 
$$ \overline{r}_1 G_1 \cup \overline{r}_2 G_2 =  {\Pi_1 \circ \Pi_2 \circ \Pi_3} {\left(G_1 \cup G_2\right)}$$
$$\textrm{with}\quad \Pi_1 = \left(\prod_{u \in G_1 \setminus G_2} (\overline{r}_1 / u) \right), \Pi_2 = \left(\prod_{u \in G_2 \setminus G_1} (\overline{r}_2 / u) \right), \Pi_3 = \left(\prod_{u \in G_1 \cap G_2}(\overline{r}_1 / u) \right).$$
Notice that $\Pi_3$ is well-defined, because we assumed the rotation sequences to be consistent so $(\overline{r}_1/u) = (\overline{r}_2/u)$.
The proof will be similar to the proof of Lemma \ref{lem:rot_comm}, as we will show that each piece of the graph ends up being the same left and right of the equality. 
We will avoid repetition of symmetric cases.
\begin{description}
\item[Edges] Without loss of generality let us consider and edge $e = (u:p, \gamma, v:q)$ in $G_1$ starting at $u:p$, the following cases are possible:
\begin{itemize}
\item if $u \in G_1 \setminus G_2$ and $u:p$ are connected, then on the left hand side we can keep track of $e$ from $G_1$ to $\overline{r}_1 G_1$, which is modified by $(\overline{r}_1 /u) \in \Pi_1$ and by $(\overline{r}_1/v)$, whilst any other rotation in $\overline{r}_1$ leaves it unchanged. On the right hand side $e$ is modified by $(\overline{r}_1/u) \in \Pi_1$ and by $(\overline{r}_1/v)$, which may be in $\Pi_1$ or $\Pi_3$---so it ends up being the same.
\item if $u \in G_1 \cap G_2$ and $u:p$ are connected only in $G_1$, then on the left hand side $e$ is modified by $(\overline{r}_1/u)$ and by $(\overline{r}_1/v)$. On the right hand side $e$ is modified by $(\overline{r}_1/u) \in \Pi_3$ and by $(\overline{r}_1/v)$, which may be in $\Pi_1$ or $\Pi_3$---so it ends up being the same.
\item if $u \in G_1 \cap G_2$ and $u:p$ are connected in both graphs, then it is clear that $e$ is modified by $(\overline{r}_1/u) \in \Pi_3$ and $(\overline{r}_1/v) \in \Pi_3$ on both sides.
\end{itemize}
\end{description}
\end{proof}

\begin{proposition}
Let $F$ be a CGD. $F$ is rotation-commuting if and only if there exists a strongly-rotation-commuting local rule $f$ which induces $F$.
\end{proposition}

\begin{proof}
$[\Leftarrow]$ Let us consider a rotation-commuting local rule $f$ of radius $r$ inducing a CGD $F$. Let $G$ be a graph and $u$ a vertex of $G$. The following sequence of equalities proves that $F$ is rotation-commuting:
\[
\begin{array}{lclr}
   F(\overline{r} G) & = &\displaystyle{\bigcup_{v\in G}} f(\overline{r} G^r_v) &\\
   & = & \displaystyle{\bigcup_{v\in G}} \overline{r}_v^* f(G^r_v) &\textrm{  (using $f$ rotation-commuting) }
\end{array}
\]
Using Lemma \ref{lem:rot_merge}, we have the existence of a rotation sequence $\overline{r'}$ such that:
$$ F(\overline{r} G)= \overline{r'} \bigcup_{v\in G} f(G^r_v)=\overline{r'} F(G) $$
\noindent $[\Rightarrow]$ Let $F$ be a rotation-commuting CGD, and $f$ a local rule inducing $F$. Informally, since $f(G_u^r)$ is included in $F(G)$ we know that as far as orientation is concerned $f$ will indeed be rotation-commuting. However it may still happen that $f$, depending upon the orientation of the vertices in $G_u^r$, will produce a smaller, or a larger, subgraph of $F(G)$. Therefore, we must define some $\tilde{f}$ which does not have this unwanted behavior. Let us consider the following function $\tilde{f}$ from $\mathcal{D}_\pi$ to $\mathcal{G}$, such that for all disk $G_u^r$:
$$ \tilde{f}(G_u^r)= \bigcup_{ \overline{r} } {\overline{r}^*}^{-1} f(\overline{r} G_u^r)  $$
with $\overline{r}^*$ a conjugate of $\overline{r}$ (given by $F$ rotation-commuting), and $\overline{r}$ ranging over the set of rotation sequences that can be applied to $G_u^r$.
\begin{itemize}
\item $\tilde{f}$ is well defined: 
by definition of $\overline{r}^*$ we have that:
\[
\begin{array}{llllr}
 &\forall  \overline{r} ,& f(\overline{r} G_u^r) \subset \overline{r}^* F(G) &\\
 \Rightarrow & \forall  \overline{r}, & {\overline{r}^*}^{-1} f( \overline{r} G_u^r)\subset F(G)& (*)\\
 \Rightarrow & \forall \overline{r}_1, \overline{r}_2 , & {\overline{r}^*_1}^{-1} f(\overline{r}_1 G_u^r)\ \textrm{and}\ {\overline{r}^*_2}^{-1} f(\overline{r}_2 G_u^r)\ \textrm{consistent}&
 \end{array}
\]
Moreover, as there is a finite number of vertices in $G_u^r$ and a finite number of possible vertex rotation, the union over all rotation sequences is a finite union of graphs.
\item $\tilde{f}$ is a local rule: we can check that it inherits of the local rule properties of $f$.
\item $\tilde{f}$ induces $F$:
\[\begin{array}{llll}
 \displaystyle{\bigcup_{v\in G}} \tilde{f}(G_u^r)& = & \bigcup_{v\in G} \left[ \displaystyle{\bigcup_{\overline{r}}} {\overline{r}^*}^{-1} f(\overline{r}G_u^r) \right]&\\
 &=& \displaystyle{\bigcup_{v\in G}} \left[ f(G_u^r) \cup \left(\displaystyle{\bigcup_{\overline{r} \neq id}} {\overline{r}^*}^{-1} f(\overline{r}G_u^r)\right)\right]&\\
 &=& F(G)  \cup \displaystyle{\bigcup_{v\in G}}\left(\displaystyle{\bigcup_{\overline{r} \neq id}} {\overline{r}^*}^{-1} f(\overline{r}G_u^r)\right)& \\
 &=F(G)& ~\textrm{since $(*)$.}
\end{array}\]
\item $\tilde{f}$ is rotation-commuting: let us consider a rotation sequence $\overline{s}$. We will show that $\tilde{f}$ commute with $\overline{s}$:
\begin{align*}
\tilde{f}(\overline{s} G_u^r )&= \bigcup_{\overline{r} } {\overline{r}^*}^{-1} f(\overline{r}\overline{s}G_u^r) \\
&= \bigcup_{\overline{r} } {\overline{s}^* \left(\overline{r}^*\overline{s}^*\right)}^{-1} f(\overline{r}\overline{s}G_u^r) \\
&= \overline{s}^* \bigcup_{\overline{r} } {\left(\overline{r}^*\overline{s}^*\right)}^{-1} f(\overline{r}\overline{s}G_u^r) \\
&= \overline{s}^* \tilde{f}(D)\textrm{, by bijection}
\end{align*}
\item $\tilde{f}$ is strongly-rotation-commuting: this is immediate from the fact that in the above formula $\overline{s}^*$ does not depends on the graph $G_u^r$ considered.
\end{itemize}
\end{proof}

\end{document}